\documentclass[envcountsame,envcountsect,10pt,oribibl]{llncs}

\usepackage{microtype}

\usepackage{enumerate}
\usepackage{comment}
\usepackage{array}
\usepackage{graphics}
\usepackage{amsmath}
\usepackage{subfig}			 %allows subfigures to be inserted in the main figure
\usepackage{tikz}
\usepackage{mathrsfs}
\usepackage{enumerate}
\usepackage{amssymb}
\usepackage{tabularx}

\newcommand{\CH}{\mathcal{H}}
\newcommand{\CC}{\mathcal{C}}

\newcommand{\pname}[1]{\textnormal{\textsc{#1}}}
\newcommand{\cclass}[1]{\textnormal{\textsf{#1}}}

\newcommand{\HED}{\pname{$H$-free Edge Deletion}}
\newcommand{\CHED}{\pname{$\CH$-free Edge Deletion}}
\newcommand{\DED}{\pname{Diamond-free Edge Deletion}}
\newcommand{\DKTED}{\pname{\{Diamond, $K_t$\}-free Edge Deletion}}

\newcommand{\SDED}{\pname{$s$-Diamond-free Edge Deletion}}
\newcommand{\SDKTED}{\pname{\{$s$-Diamond, $K_t$\}-free Edge Deletion}}
\newcommand{\SDEE}{\pname{$s$-Diamond-free Edge Editing}}
\newcommand{\SDKTEE}{\pname{\{$s$-Diamond, $K_t$\}-free Edge Editing}}

\newcommand{\HEE}{\pname{$H$-free Edge Editing}}
\newcommand{\DEE}{\pname{Diamond-free Edge Editing}}

\newcommand{\PIVD}{\pname{$\Pi$ Vertex Deletion}}
\newcommand{\KOSVD}{\pname{$K_{1,s}$-free Vertex Deletion}}
\newcommand{\KOSOVD}{\pname{$K_{1,s+1}$-free Vertex Deletion}}

\newcommand{\TSAT}{\pname{3-Sat}}

\newcommand{\VC}{\pname{Vertex Cover}}

\newcommand{\NPC}{\cclass{NP-complete}}

\newcommand{\SUBEX}{$2^{o(|G|)}$}
\newcommand{\PSUBEX}{$2^{o(k)}\cdot |G|^{O(1)}$}

\newcommand{\PH}{$\cclass{NP} \subseteq \cclass{coNP/poly}$}

\newcommand{\defstage}[2]{% PGD Version
  \hfill\\\smallskip\noindent%
  \begin{tabularx}{\textwidth}{|l X|}%
    \hline%
    \multicolumn{2}{|l|}{\textbf{#1}}\\%
    &#2\\\hline%
  \end{tabularx}%
}
\newtheorem{thm}{Theorem}[section]
\newtheorem{lem}[thm]{Lemma}

\newtheorem{pro}[thm]{Proposition}
\newtheorem{observation}[thm]{Observation}
\newtheorem{cor}[thm]{Corollary}
\newtheorem{defn}[thm]{Definition}
\newtheorem{drule}{Rule}
\newtheorem{op}{Open Problem}
\usepackage{microtype}%if unwanted, comment out or use option "draft"

\title{A cubic vertex kernel for Diamond-free Edge Deletion and more\thanks{A preliminary version
of this paper has appeared in the proceedings of IPEC 2015.}}
\author{R. B. Sandeep\inst{1}\thanks{supported by TCS Research Scholarship} \and Naveen Sivadasan\inst{2}}
\institute{Department of Computer Science \& Engineering\\
Indian Institute of Technology Hyderabad, India\\
\email{cs12p0001\makeatletter@\makeatother iith.ac.in}
\and
TCS Innovation Labs, Hyderabad, India\\
\email{naveen\makeatletter@\makeatother atc.tcs.com}}

\begin{document}

\maketitle

\begin{abstract}
  A diamond is a graph obtained by removing an edge from a 
  complete graph on four vertices. A graph is diamond-free if it does not contain an induced diamond.
  The \DED\ problem
  asks  whether there exist at most $k$
  edges in the input graph $G$ whose deletion results in a diamond-free graph. 
  For this problem, a polynomial kernel of $O(k^4$)
  vertices was found by Fellows et. al. (Discrete Optimization, 2011). 

  In this paper, we give an improved kernel of $O(k^3)$ vertices for \DED. 
  Further, we give an $O(k^2)$ vertex kernel for a related problem
  \DKTED, where $t\geq 4$ is any fixed integer. To complement our results,
  we prove that these problems are \NPC\ even for $K_4$-free graphs
  and can be solved neither in 
  subexponential time (i.e., $2^{o(|G|)}$) nor in parameterized subexponential
  time (i.e., $2^{o(k)}\cdot |G|^{O(1)}$), unless Exponential Time Hypothesis fails.
  Our reduction implies the hardness and lower bound for a general 
  class of problems, where these problems come as a special case.
\end{abstract}

\section{Introduction}

For a finite set of graphs $\CH$,
\CHED\ problem asks whether 
we can delete at most $k$ edges from an input graph $G$ 
to obtain a graph $G'$ such that for every $H\in \CH$, 
$G'$ does not have an induced
copy of $H$.
If $\CH=\{H\}$, the problem is denoted by \HED.
\pname{Editing} problems are defined similarly were we are allowed to add or delete
at most $k$ edges.
\CHED\ comes under the broader category of graph modification
problems which have found applications in DNA physical mapping 
\cite{bodlaender1996intervalizing},
circuit design \cite{el1988complexity}
and machine learning \cite{BansalBC04correlation}.
Cai has proved that \CHED\
is fixed parameter tractable~\cite{Cai96fixed}.
Polynomial kernelization and incompressibility of these problems 
were subjected to rigorous studies in the recent past.
Kratsch and Wahlstr{\"{o}}m gave the first example on the incompressibility
of \HED\ problems by proving that the problem is incompressible if $H$
is a certain graph on seven vertices, unless \PH~\cite{kratsch2013two}. 
Later, it has been proved that there exist no polynomial kernel for 
\HED\
where $H$ is any 3-connected graph other than a complete 
graph, unless \PH~\cite{CaiC15incompressibility}. In the same paper, 
under the same assumption, it is proved that, 
if $H$ is a path or a cycle, then \HED\ 
is incompressible if and only if $H$ has at least four edges. 
Except for a few cases, the kernelization complexity of \HED\
is known when $H$ is a tree~\cite{cai2012polynomial}.
It has been proved that 
\CHED\ admits
polynomial kernelization on bounded degree graphs if $\CH$ is a 
finite set of graphs \cite{drange16compressing}. 
Though kernelization complexities of 
many \HED\ problems are known, \textsc{Claw-free Edge Deletion} 
withstood the test of time and yielded neither an incompressibiltiy
result nor a polynomial kernel. Some progress has been made recently 
for this problem such as a polynomial kernel for \textsc{Claw-free Edge Deletion} 
on $K_t$-free input graphs~\cite{AravindSS14on} and a polynomial kernel for 
\textsc{$\{\text{Claw, Diamond}\}$-free Edge Deletion}~\cite{cygan2015polynomial}.

\subsection*{Motivation}
As described above, the kernelization complexity of \HED\
is known when $H$ is any 3-connected graph, path or cycle.
Except for a few cases, the status is known when $H$ is any tree.
Every new insight into these problems may help us to obtain a dichotomy
on the kernelization complexities of \HED\ problems.

The polynomial kernelization in this paper is inspired by two properties
related to diamond graph. Firstly, a graph is diamond-free if and only if 
every edge is part of exactly one maximal clique. The second property is 
that, the neighborhood of every vertex in a diamond graph is connected. 
It can be easily verified that, when a graph $H$ has this property and if 
$H$ is diamond-free then $H$ is a disjoint union of cliques (cluster). 
Though our kernelization technique may give polynomial kernels for 
\CHED, such that $\CH$ contains diamond and every other $H\in \CH$ is a cluster,
it complicates the analysis of the kernel size. Hence we restrict the study
to \DED\ and \DKTED.

\subsection*{Our Results}
In this paper, we study the polynomial kernelization and hardness results of 
\DED\ and \DKTED, where $t\geq 4$ is any fixed integer. 
It has been proved 
that \DED\ admits a kernel of $O(k^4)$ 
vertices~\cite{FellowsGKNU11graph}.
We improve this result by giving a kernel of $O(k^3)$ vertices. 
A proper subset of the rules applied for \DED\ gives us an $O(k^2)$
vertex kernel for \DKTED.
We use vertex modulator technique, which was used recently 
to give a polynomial kernel for \textsc{Trivially Perfect Editing}~\cite{DrangeP15polynomial} 
and to obtain a 
polynomial kernel for \textsc{$\{\text{Claw, Diamond}\}$-free Edge Deletion}~\cite{cygan2015polynomial}. 
We introduce a rule named as \emph{vertex-split}
which \emph{splits} a vertex $v$ into a set of independent vertices 
where each vertex in the set corresponds to a 
component in the neighborhood of $v$. 
We prove that, this rule is safe for many other \HED\ problems.

For any fixed $s\geq 1$, an $s$-diamond is defined as the graph $K_2\times (s+1)K_1$.
When $s=1$, we get a diamond graph. 
As part of a dichotomy result on the hardness of \HED\ and \HEE\ problems,
it has been proved that \SDED\ and \SDEE\ are \NPC\ and cannot be 
solved in parameterized subexponential time, unless 
Exponential Time Hypothesis (ETH) fails~\cite{AravindSS16parameterized}.
We improve these results by proving that 
\SDED\ and \SDEE\ are \NPC\ even
on $K_4$-free graphs and can be solved neither in subexponential time nor 
in parameterized subexponential time, unless ETH fails. Our reduction
implies that these results are applicable for \SDKTED\ and \SDKTEE\ for
any fixed $s\geq 1$ and $t\geq 4$.
Fellows et. al. have proved \cite{FellowsGKNU11graph} the hardness of a similar kind of problems termed as
\pname{$s$-Edge Overlap Deletion} (\pname{$s$-Edge Overlap Editing}),
where the objective is to delete (edit) at most $k$ edges from the 
input graph such that every edge in the resultant graph is part of at most
$s$ maximal cliques. We observe that when $s=1$, \SDED~(\SDEE) coincides 
with \pname{$s$-Edge Overlap Deletion} (\pname{$s$-Edge Overlap Editing}).
%It has been proved that \SDED\ and \SDEE\ are 
%\NPC\ \cite{FellowsGKNU11graph}\footnote{In \cite{FellowsGKNU11graph}, these problems are termed as
%\pname{$s$-Edge Overlap Deletion} and \text{Editing}}.
%The reduction implies that the problem cannot be solved in 
%parameterized subexponential time, unless ETH fails.
%Since the reduced instance size is not linear, their reduction
%does not imply no-subexponential time algorithm for these problems.

%\input{preliminaries}
% diamond \{a,b,c,d\} order
% G-S
% equivalent edges in G and G'
% exhaustively apply
% safety of rule
% instance means instance of \DED
% E(G) V(G)
% neighborhood
% independent set
% G[S]
% V_E
% E' forms a K_4
% an instance always means an instance of \DED
% matching, non-matching
% G-e
% S-e
% diamond-free
% solution
% maximal clique partitioning

\subsection{Preliminaries}

\noindent \textbf{Graphs:} Every graph considered here is simple, finite and undirected.
For a graph $G$, $V(G)$ and $E(G)$ denote the vertex set and the edge set of $G$ respectively.
$N_G(v)$ denotes the (open) neighborhood of a vertex $v\in V(G)$, which is the set of 
vertices adjacent to $v$ in $G$. 
%The closed neighborhood of $v$ is denoted by 
%$N_G[v]$ and is defined by $N_G(v)\cup \{v\}$. 
We remove the subscript when there is no ambiguity about the underlying graph $G$.
A graph $G'=(V',E')$ is called an induced subgraph of a graph $G$ if $V'\subseteq V(G)$, $E'\subseteq E(G)$
and an edge $\{x,y\}\in E(G)$ is in $E'$ if and only if $\{x,y\}\subseteq V'$.
For a vertex set $V'\subseteq V(G)$, $G[V']$ denotes the induced subgraph with a vertex set $V'$ of $G$.
A component $G'$ of a graph $G$ is a connected induced subgraph of $G$ such that
there is no edge between $V(G')$ and $V(G)\setminus V(G')$.
For a set of vertices $V'\subseteq V(G)$, $G-V'$ denotes the graph obtained by removing the vertices in 
$V'$ and all its incident edges from $G$. For an edge set $E'\subseteq E(G)$, $G-E'$ denotes 
the graph obtained by deleting all edges in $E'$ from $G$. 
%If $V'$ ($E'$) is a singleton set $\{v\}$ ($\{e\}$),
%we denote the graph $G-V'$ ($G-E'$) by $G-v$ ($G-e$). 
%For an edge set $E'\subseteq E(G)$,
%$V_{E'}(G)$ denotes the vertices in $G$ incident to the edges in $E'$. 
A matching (non-matching) is a set of 
edges (non-edges) such that every vertex in the graph is incident to at most one edge (non-edge) in the matching (non-matching).
$K_t$ denotes the complete graph on $t$ vertices and $K_{1,s}$
denotes the graph where a vertex is adjacent to an independent set of $s$ vertices,
i.e., the graph $K_1\times sK_1$.
An $s$-diamond is defined as the graph $K_2\times (s+1)K_1$~\cite{AravindSS16parameterized}.
We note that $1$-diamond is a diamond graph (see Figure~\ref{fig:diamond}). 
The edge between the two vertices with degree three in a diamond is the \emph{middle edge}
of the diamond.
$K_{1,3}$ is also known as a claw graph.
%A diamond is a graph obtained by deleting an edge from a complete graph on four vertices. 
In this paper, $\CH$ always denotes a finite set of graphs.
A graph $G$ is called $\CH$-free, 
if $G$ does not contain any induced copy of any $H\in \CH$.
%Every 
%Whenever we mention that $\{a,b,c,d\}\subseteq V(G)$ induces a diamond in $G$, $a$ and $b$ are degree-3 vertices
%and $c$ and $d$ are degree-2 vertices. In a diamond, we call the edge between the degree-3 vertices as the \emph{middle}
%edge.

\begin{figure}[h]
  \centering
  \includegraphics[width=1.5in]{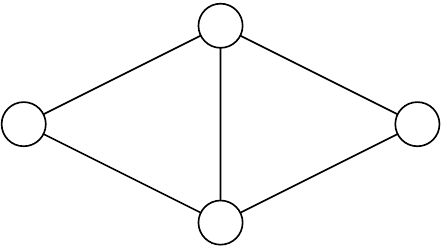}
  \caption{1-diamond is isomorphic to a diamond graph}
  \label{fig:diamond}
\end{figure}

\noindent \textbf{Parameterized complexity:} A parameterized problem is \emph{fixed parameter tractable},
if there is an algorithm to solve it in time $f(k)\cdot n^{O(1)}$, where $f$ is any computable function and $n$
is the size of the input, and $k$ is the parameter. 
A \emph{polynomial kernelization} is an algorithm which takes as input $(G,k)$
of a parameterized problem,
runs in time $(|G|+k)^{O(1)}$ and returns an instance $(G',k')$ of the same problem such that
$|G'|, k'\leq p(k)$, where $p$ is any polynomial function. A rule for kernelization is \emph{safe} if 
$(G,k)$ is a yes-instance if and only if $(G',k')$ 
is a yes-instance where $(G,k)$ and  $(G',k')$ are the input and output of 
the rule. A \emph{linear reduction} is a polynomial time reduction
from a problem $A$ to another problem $B$ such that $|G'|=O(|G|)$,
where $G$ and $G'$ are the input and output of the reduction.
A \emph{linear parameterized reduction} from a parameterized 
problem $A$ to another problem $B$ is a polynomial time reduction such that 
$k'=O(k)$ where $k$ and $k'$ are the parameters of the instances of
$A$ and $B$ respectively. 
A problem is solvable in \emph{subexponential} time
if it admits an algorithm which runs in time \SUBEX, where $G$ is the input.
Similarly, a parameterized problem is solvable in \emph{parameterized subexponential} time
if it admits an algorithm which runs in time \PSUBEX, where $G$ is the input and $k$ is the 
parameter.

Exponential Time Hypothesis (ETH) (along with Sparsification Lemma \cite{ImpagliazzoPZ01which}) 
implies that there is no algorithm which solves \TSAT\ in time $2^{o(n+m)}$, where $n$ is the number of 
variables and
$m$ is the number of clauses in the input instance. 
We can use a linear reduction from a problem 
(which does not admit subexponential time algorithm, assuming ETH) to another problem
to show that the latter does not have a subexponential time algorithm, unless ETH fails.
Similarly, under the same assumption, we can use a linear parameterized reduction from a parameterized problem
(which does not admit parameterized subexponential time algorithm, assuming ETH)
to another parameterized problem to show that the latter does not have a 
parameterized subexponential time algorithm.
We refer the book \cite{CyganFKLMPPS15parameterized} for further reading on parameterized 
algorithms and complexity.

\section{Polynomial Kernels}
In this section, we give a kernelization for \DED\ and \DKTED.
There are two phases for the kernelization. In the first phase, apart from three 
standard rules, we introduce a new rule named as vertex-split, which has applications
in the kernelization of other edge deletion problems. In the second phase,
we apply vertex modulator technique.

\subsection{Phase 1}
We start with two standard rules. 
The first rule deletes an \textit{irrelevant} edge and the second rule
deletes a \textit{must-delete} edge.
 
\begin{defn}[Core Member]
  Let $G$ be an input graph of an \CHED\ problem. 
  Then, a vertex or an edge of a graph $G$ 
  is a \emph{core member} of $G$ if it is contained in a
  subgraph (not necessarily induced) of $G$ isomorphic to an $H\in \CH$.
\end{defn}

\begin{drule}[Irrelevant Edge]
  \label{rule:irrelevant-edge}
  Let $G$ be an input to the rule, which is an input graph to an \CHED\ problem.
  If there is an edge $e\in E(G)$ which is not a core member of $G$, 
  then delete $e$ from $G$.
\end{drule}

\begin{lem}
  \label{lem:irrelevant-edge}
  Irrelevant edge rule is safe and can be applied in polynomial time for any \CHED.
\end{lem}

\begin{proof}
  Let $(G,k)$ be an instance of \CHED. 
  Let $G'$ be obtained by applying irrelevant edge rule on $G$. 
  We claim that $(G,k)$ is a yes-instance if and only if $(G',k)$ is a yes-instance.
  Let $S$ be a solution of size at most $k$ of $(G,k)$. For a contradiction, assume 
  that $G'-S$ has an induced $H\in \CH$ with a vertex set $D'$. 
  Since $D'$ does not induce $H$ in $G-S$, 
  the edge $e$ deleted by irrelevant edge rule has both the end points in $D'$. 
  Then $D'$ induces a supergraph of $H$ in $G$, which is a contradiction.
  Conversely, let $S'$
  be a solution of size at most $k$ of $(G',k)$. Assume that $G-S'$ has an induced
  $H\in \CH$ with vertex set $D$. The edge deleted by irrelevant edge rule
  has both the end points in $D$. This implies that $D$ induces a supergraph 
  of $H$ in $G$, which is a contradiction. Since, in polynomial time, we can verify whether
  an edge is part of an $H\in \CH$ in $G$, the rule can be applied in polynomial time.
\end{proof}

\begin{cor}
  \label{cor:irrelevant-edge}
  Irrelevant edge rule is safe and can be applied in polynomial time for \DED\ and \DKTED.
\end{cor}

The next rule deletes an edge $e$, if $e$ is the middle edge of $k+1$ otherwise edge-disjoint diamonds.
This rule is found in \cite{FellowsGKNU11graph}.

\begin{drule}[Sunflower]
  \label{rule:sunflower-middle}
  Let $(G,k)$ be an input to the rule. If there is an edge $e=\{x,y\}\in E(G)$ 
  such that $G[N(x)\cap N(y)]$ has a non-matching of size at
  least $k+1$,  then delete $e$ from $G$ and decrease $k$ by $1$.
\end{drule}

\begin{lem}
  \label{lem:sunflower-middle}
  Sunflower rule is safe and can be applied in polynomial time.
\end{lem}
\begin{proof}
  Let $(G,k)$ be an instance of \DED~(\DKTED).
  Let $e=\{x,y\}\in E(G)$ and $V'$ be $N(x)\cap N(y)$. Assume that $G[V']$ 
  has a non-matching $M'$ of size at
  least $k+1$. Let sunflower rule be applied on $(G,k)$ to obtain $(G-e,k-1)$.
  It is enough to prove that every solution $S$ of size at most $k$ of $(G,k)$ 
  contains the edge $e$.
  Every non-edge $\{a,b\}$ in $M'$ corresponds to 
  an induced diamond $\{x,y,a,b\}$ in $G$. 
  The diamonds corresponds to any two different non-edges in $M'$ share only
  one edge $\{x,y\}$. Since at least one edge from every induced diamond is in $S$, 
  $e$ must be in $S$. 
  The rule can be applied in 
  polynomial time as maximum non-matching can be found in polynomial-time.
\end{proof}

Now we introduce a property and a rule based on it. 

\begin{defn}[Connected Neighborhood]
  For a graph $G$ and a vertex $v\in V(G)$, $v$ has \emph{connected neighborhood} if $G[N(v)]$ is connected.
  $G$ has connected neighborhood if every vertex in $G$ has connected neighborhood.
\end{defn}

\begin{drule}[Vertex-Split]
  \label{rule:vertex-split}
  Let $v\in V(G)$ and $v$ does not have connected neighborhood in $G$. 
  Let there be $t>1$ components in $G[N(v)]$ with vertex sets 
  $V_1, V_2,\ldots,V_t$. Introduce $t$ new vertices $v_1, v_2,\ldots,v_t$ and make
  $v_i$ adjacent to all vertices in $V_i$ for $1\leq i\leq t$. Delete $v$.
\end{drule}

An example of the application of vertex-split rule is depicted in Figure~\ref{fig:vs}.
We denote the set of vertices created by splitting $v$ by $V_v$.
Let $G'$ be the graph obtained by splitting a vertex $v$ in $G$.
For convenience, we identify an edge $(v, u)$ in $G$ with an edge $(v_j, u)$ in $G'$
where $u$ is in the $j^{th}$ component of $G[N(v)]$, so that for every set of 
edges $S$ in $G$, there is a corresponding set of edges in $G'$ and vice versa. 
We identify a set of vertices $V'\subseteq V(G)\setminus\{v\}$ with the corresponding vertices in $G'$.
Similarly, we identify $V'\subseteq V(G')\setminus V_v$ with the corresponding vertex set in $G$.
Before proving the safety of the rule, we prove two simple observations.

\begin{figure}[h]
\centering
\subfloat[Subfigure 1 list of figures text][A graph $G$ where a vertex $v$ has a disconnected neighborhood.]{
\resizebox{3cm}{!}{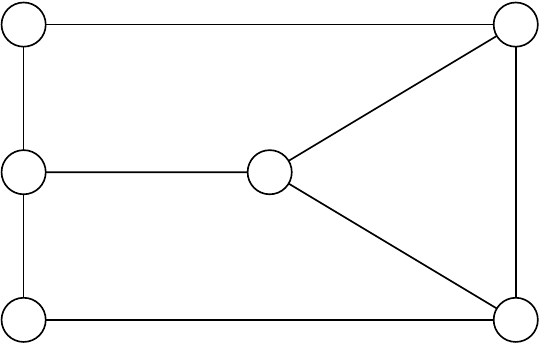}
\label{fig:vs1}}
%\caption{}
\qquad
\subfloat[Subfigure 2 list of figures text][Vertex-split rule is applied at $v$.]{
\resizebox{3cm}{!}{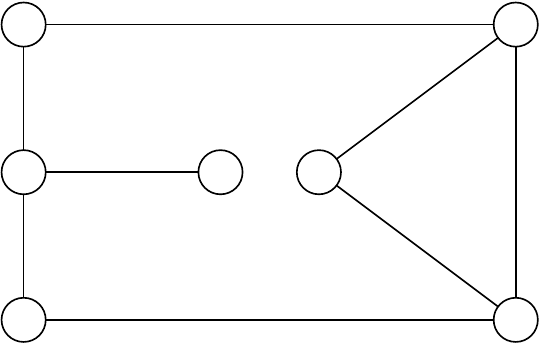}
\label{fig:vs2}}
\caption{An application of vertex-split rule}
\label{fig:vs}
\end{figure}

\begin{observation}
  \label{obs:vertex-split}
  Let vertex-split rule be applied on $G$ to obtain $G'$. 
  Let $v\in V(G)$ be the vertex being split. 
  \begin{enumerate}[(i)]
  \item\label{item:obs:vertex-split-d4} Then, 
    for every pair of vertices $\{v_i,v_j\}\subseteq V_v$, the distance
    between $v_i$ and $v_j$ is at least four.
  \item\label{item:obs:vertex-split-cn}
    Let $u\in V(G)\setminus \{v\}$ and $u$ has connected neighborhood in $G$. 
    Then $u$ has connected neighborhood in $G'$ . Furthermore, every new vertex $v_i$ 
    introduced in $G'$ has connected neighborhood.
  \end{enumerate}
\end{observation}

\begin{proof}
  (\ref{item:obs:vertex-split-d4}).
  Let $\{v_i,v_j\}\subseteq V_v$. Clearly, $v_i$ and $v_j$ are non-adjacent. Consider any two vertices $u_i\in N(v_i)$ and
  $u_j\in N(v_j)$. If $u_i= u_j$ or $u_i$ and $u_j$ are adjacent in $G'$, there would be only one vertex
  generated for the component containing $u_i$ and $u_j$ in $G[N(v)]$ by splitting $v$, which is a contradiction.
  It follows that
  the distance between $v_i$ and $v_j$ is at least four.

  (\ref{item:obs:vertex-split-cn}).
  If $v\notin N_G(u)$,
  then the neighborhood of $u$ is not disturbed by the rule and hence
  $u$ has connected neighborhood
  in $G'$. Let $v\in N(u)$. 
  Let $v_i$ be the vertex generated by splitting $v$ for the component in $G[N(v)]$ containing $u$.
  Since, there is only one new vertex introduced for a component of $G[N(v)]$, no other new vertex is adjacent to $u$ in $G'$. 
  It is straight-forward to verify that $v_i$ in $G'[N_{G'}(u)]$ plays the role of $v$ in
  $N_{G}(u)$ and hence $G[N_G(u)]$ and $G'[N_{G'}(u)]$ are isomorphic.
  To prove the last statement, we observe that, 
  since a new vertex is made adjacent to a component in the neighborhood of $v$, 
  every new vertex $v_j$ in $G'$ has connected neighborhood.
\end{proof}

\begin{lem}
  \label{lem:vertex-split}
  Vertex-split rule is safe and can be applied in polynomial time for any \pname{$\mathcal{H}$-free Edge Deletion}
  problem where every $H\in \mathcal{H}$ has diameter at most two and has 
  connected neighborhood.
\end{lem}

\begin{proof}
  Let $G'$ be obtained by applying vertex-split rule on a vertex $v$ of $G$. We claim 
  that $(G,k)$ is a yes-instance if and only if $(G',k)$ is a yes-instance.

  Let $(G,k)$ be a yes-instance. Let $S$ be a solution of size at most $k$ of $(G,k)$.
  For a contradiction, assume that $G'-S$ has an induced $H\in \mathcal{H}$ with a vertex set $D'$.
  Since $G-S$ is $H$-free, $D'$ must contain at least one newly created vertex $v_i$. 
  Since the diameter of $H$ is at most two, by 
  Observation~\ref{obs:vertex-split}(\ref{item:obs:vertex-split-d4}), $D'$ can contain at most
  one newly created vertex. 
  Hence, let $D'\cap V_v = \{v_i\}$. 
  Since $H$ is connected, $D'\cap N(v)\neq \emptyset$.
  Now, there are two cases and in each case we get a contradiction.
  \begin{enumerate}[(i)]
    \item 
      $D'\cap N(v)\subseteq V_i$: In this case, $v_i$ plays the role of $v$ and hence
      $(G-S)[(D'\setminus v_i)\cup \{v\}]$ and $(G'-S)[D']$ are isomorphic.
    \item 
      $D'$ contains vertices from multiple components of $N(v)$, i.e., 
      $(D'\cap N(v))\setminus V_i\neq \emptyset$: 
      Let $u_j\in D'\cap (N(v)\setminus V_i)$.
      Now, it is straight-forward to verify that the distance between $v_i$ and $u_j$
      is at least three in $(G'-S)[D']$, which is a contradiction to the fact that $H$ has 
      diameter at most two.
  \end{enumerate}  
  For the converse, let $S'$ be a solution of size at most $k$ of $(G',k)$.
  For a contradiction, assume that $G-S'$ has an induced $H\in \mathcal{H}$
  with a vertex set $D$. Clearly, $v\in D$. Since $H$ is connected, there 
  exists a $u_i\in D\cap N(v)$.
  Now there are two cases and in each case we get a contradiction.
  \begin{enumerate}[(i)]
    \item 
      $D\cap N(v)\subseteq V_i$: In this case, $(D\setminus \{v\})\cup \{v_i\}$
      induces $H$ in $G'-S'$, which is a contradiction. 
    \item 
      $D$ contains vertices from multiple components of $N(v)$:
      Let $u_j\in (D\cap N(v))\setminus V_i$.
      Here, it can be verified that either $v$ does not have connected neighborhood
      in $(G-S')[D]$ or the distance between either $v$ and $u_i$ or between 
      $v$ and $u_j$ is at least three, which is a contradiction.
  \end{enumerate}  
  It is straight-forward to verify that splitting a vertex can be done in linear time.
\end{proof}

\begin{cor}
  \label{cor:vertex-split}
  Vertex-split rule is safe and can be applied in polynomial time for the 
  problems \DED\ and \DKTED.
\end{cor}

The next rule is a trivial one and the safety of it can be easily verified.

\begin{drule}[Irrelevant component]
  \label{rule:irrelevant-component}
  Let $G$ be an input to the rule, which is an input graph of an $\CHED$ problem.
  If a component of $G$ is $\CH$-free, then delete the component from $G$.
\end{drule}

\begin{lem}
  \label{lem:irrelevant-component}
  Irrelevant component rule is safe and can be applied in polynomial time for every \CHED\ problem.
\end{lem}

\begin{cor}
  \label{cor:irrelevant-component}
  Irrelevant component rule is safe and can be applied in polynomial time for \DED\ and \DKTED.
\end{cor}

Now, we are ready with the Phase 1 of the kernelization for \DED\ and \DKTED.

\defstage{Phase 1}
{ Let $(G,k)$ be an input to Phase 1.
  \begin{itemize}
  \item Exhaustively apply rules irrelevant edge,  sunflower, vertex split  and irrelevant component
    on $(G,k)$ to obtain $(G',k)$.
  \end{itemize}
}

\begin{lem}
  \label{lem:phase1-props}
  Let $(G,k)$ be an instance of \DED~(\DKTED).
  Let $(G',k')$ be obtained by applying Phase 1 on $(G,k)$. Then:
  \begin{enumerate}[(i)]
    \item \label{item:phase1-prop-cm} Every vertex and edge in $G$ is a core member.
    \item \label{item:phase1-prop-nh} $G'$ has connected neighborhood.
%    \item \label{item:phase1-prop-id} Every component in $G'$ has an induced diamond (diamond or $K_t$).
    \item \label{item:phase1-prop-bound} $|E(G')|\leq |E(G)|$ and $|V(G')|\leq 2|E(G)|$.
  \end{enumerate}
\end{lem} 

\begin{proof}
  (\ref{item:phase1-prop-cm}) and (\ref{item:phase1-prop-nh}) %and (\ref{label:phase1-prop-id})
  follow from the fact that irrelevant edge and vertex-split % and irrelevant component 
  rules are not applicable on $(G',k')$.
  (\ref{item:phase1-prop-bound}) follows from the fact that none of the rules 
  increases the number of edges in the graph.

\end{proof}

\begin{lem}
  \label{lem:phase1-safe}
  Applying Phase 1 is safe and Phase 1 runs in polynomial time.
\end{lem}
\begin{proof}
  The safety follows from the safety of the rules being applied. 
  Single application of each rule can be done in polynomial time
  (Corollary~\ref{cor:irrelevant-edge}, Lemma~\ref{lem:sunflower-middle}, 
  Corollary~\ref{cor:vertex-split} and Corollary~\ref{cor:irrelevant-component}). 
  None of the rules increases the number of edges. Hence, number of applications of 
  irrelevant edge rule and sunflower rule is linear. An application of irrelevant component
  rule does not necessitate an application of vertex-split rule.
  By Observation~\ref{obs:vertex-split}(\ref{item:obs:vertex-split-cn}), 
  an application of vertex-split rule decreases the number of vertices with
  disconnected neighborhood. Hence, between two applications of either 
  irrelevant edge rule or sunflower rule, only a linear number of applications 
  of vertex-split rule is possible. Hence, vertex-split rule can be applied
  only polynomial number of times.
  Since, only vertex-split rule increases the number of vertices and there are 
  only polynomial many applications of it, irrelevant-component rule can be 
  applied only polynomial number of times.
\end{proof}
%%%%%%%%%%%%%%%%%%%%%%%%%%%%%%%%%%%%%%%%%%%%%%%%%%%%%%%%%%%%%%%%%%%%%%%%%%%%%%%%%%%%%%%%%%%%%%%%%

%\input{kernel2}
\subsection{Phase 2}
In this phase, we apply vertex modulator technique to complete
the kernelization of \DED\ and \DKTED.
We define a vertex modulator for \DED~(\DKTED) similar to that defined 
for \textsc{Trivially Perfect Editing}~\cite{DrangeP15polynomial}.

\begin{defn}[D-modulator]
  \label{def:modulator}
  Let $(G,k)$ be an instance of \DED~(\DKTED). 
  Let $V'\subseteq V(G)$ be such that $G - V'$ is diamond-free ($\{\text{diamond},K_t\}$-free).
  Then, $V'$ is called a D-modulator.
\end{defn}

Now we state a folklore characterization of diamond-free graphs.
\begin{pro}%[\cite{farrugia2002clique}]
  \label{pro:diamond-free}
  A graph $G$ is diamond-free if and only 
  if every edge in $G$ is a part of exactly one maximal clique.
\end{pro}

For a diamond-free graph $G$, since every edge is in exactly one maximal clique, 
there is a unique way of partitioning the edges 
into maximal cliques. For convenience, we call the set of 
subsets of vertices, where each subset is the vertex set of a maximal clique,
as a \emph{maximal clique partitioning}. 
We note that, one vertex may be a part of many sets in the partitioning.

\begin{lem}
  \label{lem:modulator}
  Let $(G,k)$ be an instance of \DED~(\DKTED). 
  Then, in polynomial time, the edge set $X$ of size at most $5k$~($t\cdot(t-1)\cdot k/2$) of a maximal 
  set of edge-disjoint diamonds (diamonds and $K_t$s), a D-modulator $V_X$ of size at most $4k$~($tk$) 
  and a maximal clique partitioning $\mathcal{C}$
  of $G-V_X$ can be obtained or 
  it can be declared that $(G,k)$ is a no-instance.
\end{lem}

\begin{proof}
  We prove the lemma for \DED. Similar arguments apply for \DKTED.
  Let $X=\emptyset$. Include edges of any induced diamond of $G$ in $X$. 
  Then, iteratively include edges of any induced diamond of $G - X$ in $X$ 
  until $k+1$ iterations are completed or no more induced diamond is found in $G - X$. 
  If $k+1$ iterations are completed, then we can declare that the instance is a 
  no-instance as every solution must have at least one edge from every induced diamonds.
  If the number of iterations is less than $k+1$ such that there is no
  induced diamond in $G - X$, then $|X|\leq 5k$, as every diamond has five edges.
  Let $V_X$ be the 
  set of vertices incident to the edges in $X$. Then $|V_X|\leq 4k$, as every diamond has four vertices.
  Since $G - V_X$ has no induced diamond,
  $V_X$ is a D-modulator.  
  Since, there are only at most $k+1$ iterations and each iteration takes polynomial time, 
  this can be done in polynomial time. Since $G - V_X$ is diamond-free, by Proposition~\ref{pro:diamond-free},
  every edge in it is part of exactly one maximal clique.
  Now, the maximal clique partitioning $\mathcal{C}$ of $G - V_X$
  where each $C\in \mathcal{C}$ is a set of vertices of a maximal clique, 
  can be found by greedily obtaining the maximal cliques,
  which can be done in polynomial time.
\end{proof}

Let $(G,k)$ be an output of Phase 1. Here onward, we assume that $X$ is an edge set of 
the maximal set of edge-disjoint diamonds (diamonds and $K_t$s), 
$V_X$ is a D-modulator, which is the set of vertices incident to $X$ and 
$\mathcal{C}$ is the maximal clique partitioning of $G - V_X$.
Observation~\ref{obs:modulator-intersect} directly follows from the maximality of $X$.
Observation~\ref{obs:intersection} is found in Lemma 3.1 of \cite{cygan2015polynomial}. It was proved there,
if $G$ is $\{\text{claw, diamond}\}$-free, but is also applicable if $G$ is diamond-free.

\begin{observation}
  \label{obs:modulator-intersect}
  Every induced diamond (diamond and $K_t$) in $G$ has an edge in $X$.
\end{observation}

\begin{observation}
  \label{obs:intersection}
  Let $C,C'\in \mathcal{C}$ and be distinct. Then:
  \begin{enumerate}[(i)]
    \item\label{label:obs-intersection-1} $|C\cap C'|\leq 1$.
    \item\label{label:obs-intersection-noedge} If $v\in C\cap C'$, then there is no edge between $C\setminus \{v\}$ and $C'\setminus \{v\}$.
  \end{enumerate}
\end{observation}

\begin{proof}
  (\ref{label:obs-intersection-1}).
  Assume that $x,y\in C\cap C'$. Then the edge $\{x,y\}$ is part of two maximal cliques, which is a 
  contradiction by Proposition~\ref{pro:diamond-free}.

  (\ref{label:obs-intersection-noedge}).
  Let $x\in C\setminus \{v\}$ and $y\in C'\setminus \{v\}$. Let $x$ and $y$ be adjacent. 
  Clearly, $\{x,y\}$ is not part of the clique induced by $C$. Now, $\{x,v\}$
  is part of not only the clique induced by $C$ but also a maximal clique containing
  $x,y$ and $v$, which is a contradiction.
\end{proof}

\begin{defn}[Local Vertex]
  \label{def:local}
  Let $G$ be a graph and $C\subseteq V(G)$ induces a clique in $G$. A vertex $v$ in $C$
  is called local to $C$ in $G$, if $N(v)\subseteq C$. 
\end{defn}

\begin{lem}
  \label{lemma:clique-reduce}
  Let $(G,k)$ be an instance of \DED~(\DKTED). Let $C$ be a clique with at least $2k+2$ vertices in $G$.
  \begin{enumerate}[(i)]
    \item\label{label:clique-reduce-size}
      Every solution $S$ of size at most $k$ of $(G,k)$ does not contain any edge $e$ where both 
      the end points of $e$ are in $C$.
    \item\label{label:clique-reduce-nodiamond}
      Let $C'\subseteq C$ be such that every vertex $v\in C'$ is local to $C$ in $G$. 
      Every induced diamond with vertex set $D$ in $G$ can contain at most one vertex in $C'$.
  \end{enumerate}
\end{lem}

\begin{proof}
  (\ref{label:clique-reduce-size}).
  Let $e=\{x,y\}$ be an edge in $G$ such that $x,y\in C$. 
  Let $S$ be a solution of size at most $k$ of $(G,k)$ such that $e\in S$.
  Consider any two vertices $a,b\in C\setminus \{x,y\}$ (assuming $k$ is at least 1). Clearly, $\{a,b,x,y\}$
  induces a diamond in $G-e$. Consider a maximum matching $M$ of $G[C\setminus \{x,y\}]$.
  Since $C\setminus \{x,y\}$ induces a clique of size at least $2k$ in $G$, $|M|\geq k$.
  For any two edges $\{a,b\}, \{a',b'\} \in M$,
  the diamonds induced by $\{a,b,x,y\}$ and $\{a',b',x,y\}$
  are edge-disjoint. $S$ must contain one 
  edge from the diamonds corresponds each edge in $M$. Since $e\in S$, $|S|\geq k+1$,
  which is a contradiction.

  (\ref{label:clique-reduce-nodiamond}).
  For a contradiction, assume that $D$ induces a diamond in $G$ and $D$ contains two 
  vertices $\{x,y\}$ of $C'$. Let $a$ and $b$ be the other two vertices in $D$.
  Since $x$ and $y$ are local to $C$ in $G$, $a,b\in C$. Hence, $\{a,b,x,y\}$
  is a clique in $G$, which is a contradiction. 
\end{proof}

We note that the following lemma is applicable only for \DED.
\begin{lem}
  \label{lem:onlyded}
  Let $C'\subseteq C$ be such that every vertex $v\in C'$ is local to $C$ in $G$. Then, 
  for the problem \DED, it is safe
  to delete $\min\{|C'|-1, |C|-(2k+2)\}$ vertices of $C'$ in $G$.
\end{lem}
\begin{proof}
  Let $G'$ be obtained by deleting a set $C''$ of $t$ vertices of $C'$ from $G$ such that 
  $t = \min\{|C'|-1, |C|-(2k+2)\}$.
  We need to prove that $(G,k)$ is a yes-instance if and only if $(G',k)$ is a yes-instance.
  Let $S$ be a solution of size at most $k$ of $(G,k)$. 
  Since $G'-S$ is an induced subgraph of $G-S$, and $G-S$ is diamond-free, we obtain that
  $G'-S$ is diamond-free.
  Conversely, let $S'$ be a solution of size at most $k$ of $(G',k)$. 
  We claim that $S'$ is a solution of $(G,k)$. Assume not.
  Let $G-S'$ has an induced diamond with a vertex set $D$. Since $|C|-|C''|\geq 2k+2$,
  by \ref{lemma:clique-reduce}(\ref{label:clique-reduce-size}),  
  $S'$ does not contain any edge in the clique induced by $C\setminus C''$ in $G'$. 
  Now there are three cases:
  \begin{enumerate}[(a).]
  \item $C''\cap D=\emptyset$: In this case $D$ induces a diamond in $G'-S'$, which is a contradiction.
  \item $C''\cap D=\{v\}$: We observe that we retained at least one vertex $u$ of $C'$ in $G'$. 
    By \ref{lemma:clique-reduce}(\ref{label:clique-reduce-nodiamond}), 
    $D$ does not contain any other vertex from $C'$.
    Then, $D\cup \{u\}\setminus \{v\}$ induces a diamond in $G'-S'$.
  \item $|C''\cap D|\geq 2$: This case is not possible by 
    \ref{lemma:clique-reduce}(\ref{label:clique-reduce-nodiamond}).
  \end{enumerate}
\end{proof}

We define $\CC_i\subseteq \CC$ as the set of sets of vertices of the maximal cliques with 
exactly $i$ vertices. Similarly, $\CC_{\geq i}\subseteq \CC$ denote the set of sets of 
vertices of the maximal cliques with at least $i$ vertices.

%We partition $\mathcal{C}$ into three - $\mathcal{C}_1, \mathcal{C}_2$ and 
%$\mathcal{C}_{\geq 3}$, the sets of vertices of maximal cliques with one, 
%two and three or more vertices respectively. 
The first in the following observation has been proved in Lemma 3.2 
in \cite{cygan2015polynomial} in the context where $G - V_X$ is $\{\text{diamond, claw}\}$-free. 
Here we prove it in the context where $G - V_X$ is diamond-free.

\begin{observation}
  \label{obs:abd}
  Let $C\in \mathcal{C}$. Then:
  \begin{enumerate}[(i)]
    \item\label{label:obs-abd-ac}
      If there is a vertex $v\in V_X$ such that $v$ is adjacent to at least two vertices in $C$,
      then $v$ is adjacent to all vertices in $C$.
    \item\label{label:obs-abd-bc} A vertex in $V(G)\setminus (V_X\cup C)$ is adjacent to at most one vertex in $C$.
  \end{enumerate}
\end{observation}

\begin{proof}
  (\ref{label:obs-abd-ac}).
  Let $v$ is adjacent to two vertices in $x,y$ in $C$ but not adjacent to $z\in C$. Then $\{x,y,v,z\}$ induces a
  diamond such that none of the edges of the diamond is in $X$.

  (\ref{label:obs-abd-bc}).
  Assume that a vertex $u\in V(G)\setminus (V_X\cup C)$ 
  is adjacent to all vertices in $C$. This contradicts with the fact that $C$ induces a maximal clique in $G - V_X$. Let $u$
  be adjacent to at least two vertices $\{a,b\}$ in $C$ and non-adjacent to at least one vertex $v\in C$.
  Then $\{a,b,u,v\}$ induces a diamond where none of the edges of the diamond is in $X$.
\end{proof}

Consider $C\in \mathcal{C}$. We define three sets of vertices in $G$ based on $C$. 
\begin{eqnarray*}
A_C &=& \{v\in V_X: v\ \text{is adjacent\ to\ all\ vertices\ in}\ C\}\\
B_C &=& \{v\in V(G)\setminus (V_X\cup C): v\ \text{is adjacent\ to\ exactly\ one\ vertex\ in}\ C\}\\
D_C &=& \{v\in V_X: v\ \text{is adjacent\ to\ exactly\ one\ vertex\ in}\ C\}
\end{eqnarray*}
For a vertex $v\in C$, let $B_v$ denote the set of all vertices in $B_C$ adjacent to $v$. 
Similarly let $D_v$ denote the set of all vertices in $D_C$ adjacent to $v$.

\begin{observation}
  \label{obs:abdmore}
  Let $C\in \mathcal{C}$. Then,
  \begin{enumerate}[(i)]
  \item \label{label:obs-abdmore-all}
    The set of vertices in $V(G)\setminus C$ adjacent to at least one vertex in $C$ is $A_C\cup B_C\cup D_C$.
  \item \label{label:obs-abdmore-ac}
    If $|C|>1$, then $A_C$ induces a clique in $G$.
  \item \label{label:obs-abdmore-disjoint}
    For two vertices $u,v\in C$, $B_u\cap B_v=\emptyset$ and $D_u\cap D_v=\emptyset$.
  \end{enumerate}
\end{observation}

\begin{proof}
  (\ref{label:obs-abdmore-all}) directly follows from Observation~\ref{obs:abd}.

  (\ref{label:obs-abdmore-ac}).
  Assume not. Let $a$ and $b$ be two non-adjacent vertices in $A_C$. 
  By Observation~\ref{obs:abd}(\ref{label:obs-abd-ac}), both $a$ and $b$ are 
  adjacent to all vertices in $C$. Consider any two vertices $x,y\in C$.
  $\{x,y,a,b\}$ induces a diamond with no edge in $X$, which is a contradiction.

  (\ref{label:obs-abdmore-disjoint}) directly follows from the definition of $B_C$ and $D_C$.
\end{proof}

\begin{lem}
  \label{lemma:abd-connected}
  Let $v\in C\in \CC$. If $B_v$ is non-empty then $D_v$ is non-empty.
\end{lem} 

\begin{proof}
  The statement is trivially true if $|C|=1$. Hence assume that $|C|\geq 2$.
  Since $v$ has connected neighborhood, $G[N(v)]$ is connected. We observe that
  $N(v)=A_C\cup B_v\cup D_v\cup (C\setminus \{v\})$. Assume $B_v$ is non-empty.
  By Observation~\ref{obs:abd}(\ref{label:obs-abd-bc}), there is no edge between the sets $B_v$ and $C\setminus \{v\}$. 
  Consider a vertex $v_b\in B_v$ adjacent to $A_C\cup D_v$. Assume $v_b$ is not adjacent to 
  $D_v$. Then $v_b$ must be adjacent to a vertex $v_a\in A_C$. Let $v'$ be any other vertex in $C$.
  Then $\{v_a,v,v',v_b\}$ induces a diamond which has no edge intersection with $X$.
  Therefore $v_b$ must be adjacent to a vertex in $D_v$.
\end{proof}

\begin{observation}
  \label{obs:twox}
  Let $C\in \mathcal{C}$. Then there are two adjacent vertices $x$ and $y$ such that $x\in A_C$ and $y\in A_C\cup D_C$.
\end{observation}

\begin{proof}
\textbf{Case 1:} $C=\{v\}\in \mathcal{C}_1$. Since $\{v\}\in \mathcal{C}_1$, $v$ is 
      not adjacent to any vertex in $V(G)\setminus V_X$. Since $v$ is a core member,
      $v$ is part of an induced diamond or $K_4$ in $G$. Hence there exist two adjacent vertices 
      $x,y\in A_C$.

\begin{comment}
\textbf{Case 2:} $C=\{u,v\}\in \mathcal{C}_2$. Since the edge $\{u,v\}$ is a core member,
      it is part of some induced diamond or $K_4$ in $G$. Let $a,b$ be the other two
      vertices in an induced diamond or $K_4$ in which $\{u,v\}$ is a part. If both $a,b\in V(G)\setminus V_X$,
      then it contradicts with either the maximality of $X$ (if $a,b,u$ and $v$ induce a diamond) or with
      the fact that $\{u,v\}$ is part of exactly one maximal clique $C$ (if $a,b,u$ and $v$ induce a $K_4$). 
      Let $a\in V_X$ and $b\in V(G)\setminus V_X$.
      Then, if $a,b,u$ and $v$ induces a diamond, then it contradicts with the maximality of $X$.
      If $a,b,u$ and $v$ induces a $K_4$, then $u,v$ and $b$ induce a $K_3$ which contradicts with the 
      fact that $\{u,v\}$ is a part of exactly one maximal clique. Hence $a,b\in V_X$. Since $a,b,u,v$
      induce a diamond or a $K_4$, one of $a,b$ must be adjacent to both $u$ and $v$ and the other 
      vertex must be adjacent to at least one of $u$ and $v$.
\end{comment}
\textbf{Case 2:} $|C|\geq 2$. 
      Assume that $|A_C|=0$. If $B_C\cup D_C = \emptyset$, 
      then by Observation~\ref{obs:abdmore}(\ref{label:obs-abdmore-all}), the clique $C$
      is a component in $G$. Then, irrelevant component rule is applicable.
      Hence $B_C\cup D_C$ is non-empty. 
      Consider a vertex $v\in C$ such that $B_v\cup D_v$ is non-empty. 
      By Observation~\ref{obs:abdmore}(\ref{label:obs-abdmore-disjoint}), $B_v\cup D_v$
      is not adjacent to any vertex in $C\setminus \{v\}$. Hence, $G[N(v)]$ has 
      at least two components, one from $B_v\cup D_v$ and the other from $C$, which contradicts with 
      the fact that $v$ has connected neighborhood. 
      Hence, $|A_C|>0$. Assume $|A_C=\{x\}|=1$. For a contradiction, assume that $D_C=\emptyset$. Then Lemma~\ref{lemma:abd-connected}
      implies that $B_C$ is empty. Then $x$ does not have connected neighborhood or $C\cup \{x\}$ induces an irrelevant 
      component, which are contradictions. 
      Hence, $D_C$ is non-empty. 
      If $|A_C|\geq 2$, then we are done by Observation~\ref{obs:abdmore}(\ref{label:obs-abdmore-ac}).
\end{proof}

%We need the following lemma, only in the case of \DED.
\begin{lem}
  \label{lemma:manylocal}
  In the context of \DED, let $C\in \mathcal{C}_{\geq 3}$. Then, the number of vertices 
  in $C$ which are adjacent to at least one vertex in $B_C\cup D_C$ is at most $4k-1$.
\end{lem}

\begin{proof}
%  We prove the statement in the context of \DED. Similar arguments apply for \DKTED.
  By Observation~\ref{obs:twox}, $|A_C|\geq 1$.
  Since $|V_X|\leq 4k$, $|D_C|\leq 4k-1$. Let $C'$ be the set of vertices in $C$ which are 
  adjacent to $B_C\cup D_C$. For every vertex $v\in C'$, by Lemma~\ref{lemma:abd-connected}, if $B_v$
  is non-empty, then $D_v$ is non-empty. Since $v\in C'$, if $B_v$ is empty, then also $D_v$ is non-empty.
  For any two vertices $v,u\in C'$, by Observation~\ref{obs:abdmore}(\ref{label:obs-abdmore-disjoint}), 
  $D_u\cap D_v=\emptyset$. Therefore $|C'|\leq |D_C|\leq 4k-1$.
\end{proof}

Now, we state the last rule of the kernelization. We apply this rule 
only for \DED.

\begin{drule}[Clique Reduction]
  \label{rule:clique-reduce}
  Let $C\in \mathcal{C}_{\geq 3}$ be such that $|C|>4k$. Let $C'$ be $C\cup A_C$. Let $C''$ be the set of vertices in $C$
  which are local to $C'$. Then, delete any $|C''|-1$ vertices from $C''$.
\end{drule}

Clique reduction rule helps us to reduce the size of large
cliques in the clique partitioning in the context of \DED. 
This rule is not required for \DKTED\ as the size of the cliques in the clique partitioning
is already bounded. 

\begin{observation}
  \label{obs:clique-bound}
  After the application of clique reduction rule, the number of vertices retained in $C$ is at most $4k$. 
\end{observation}

\begin{proof}
  By Lemma~\ref{lemma:manylocal}, the number of vertices in $C$ which are not local to $C'$ is at most $4k-1$.
  Hence, the rest of the vertices in $C$ are local to $C'$ in $G$. If $|C|>4k$, clique reduction rule
  retains only one local vertex and delete all other vertices in $C$ local to $C'$. 
\end{proof}

\begin{lem}
  \label{lem:rule-clique-reduce-safe}
  Clique reduction rule is safe for \DED\ and can be applied in polynomial time.
\end{lem}
\begin{proof}
  The safety of the rule follows from Lemma~\ref{lem:onlyded}. It is straight-forward
  to verify that the rule can be applied in polynomial-time.
\end{proof}

Now we give the kernelization algorithms.

\defstage{Kernelization of \DED}
{ Let $(G,k)$ be the input.
  \begin{enumerate}[Step 1:]
    \item\label{item:kernelization-phase1} Apply Phase 1 on $(G,k)$ to obtain $(G_1,k')$.
    \item\label{item:kernelization-rule-final} Greedily pack edge disjoint diamonds of $(G_1,k')$. 
      If the count of edge disjoint diamonds in the pack exceeds $k$, then declare that the 
      instance is a no-instance; Otherwise find $X, V_X$ and $\mathcal{C}$ 
      of $G_1$ as given in Lemma~\ref{lem:modulator} from the maximal greedy packing. 
    \item\label{item:kernelization-rule-cr} Exhaustively apply clique reduction rule on $(G_1,k')$ to obtain $(G',k_1)$.
%    \item\label{label:kernelizatoin-final} If neither Step~\ref{label:kernelization-phase1} nor 
%      Step~\ref{label:kernelization-rule-final} is applicable on $(G',k_1)$, then return $(G',k_1)$. 
%      Otherwise apply the kernelization on $(G',k_1)$.
  \end{enumerate}
}
\defstage{Kernelization of \DKTED}
{ Let $(G,k)$ be the input. Apply Phase 1 on $(G,k)$ and return the output $(G',k')$.
  If the size of a greedy packing of edge disjoint diamonds and $K_t$s of $G'$ exceeds $k$,
  then declare that the instance is a no-instance; Otherwise return $(G',k')$.
}
\begin{lem}
  \label{lemma:kernelizaton-safe}
  The kernelization algorithms for \DED\ and \DKTED\ are safe 
  and can be applied in polynomial time.
\end{lem}

\begin{proof}
  The safety of the kernelizations follow directly from the safety of 
  Phase 1 and clique reduction rule (Lemma~\ref{lem:phase1-safe}, \ref{lem:rule-clique-reduce-safe}).  
  Every application of clique reduction rule
  decreases the number of edges. 
%  Hence Step~\ref{label:kernelization-rule-final} runs at most $|E(G)|$ times. 
  Hence by Lemma~\ref{lem:phase1-safe}, \ref{lem:modulator} and \ref{lem:rule-clique-reduce-safe}
  both the kernelization runs in polynomial time.
%  By Lemma~\ref{lem:phase1-safe}, Phase 1 runs in polynomial time. 
%  This proves that the kernelization of \DKTED\ runs in polynomial time. 
%  By Lemma~\ref{lemma:modulator}, 
%  and Lemma~\ref{lemma:rule-clique-reduce-safe},
%  Step~\ref{label:kernelization-rule-final} runs in polynomial time. 
%  Two execution of Phase 1 cannot be done consecutively. Hence the kernelization runs in polynomial time.
\end{proof}

\subsection{Bounding the Kernel Size}
In this subsection, we bound the number of vertices in the kernels obtained by the kernelizations.
Let $(G,k)$ be an instance of \DED~(\DKTED) and $(G',k')$ is obtained by the kernelization. 
Consider an $X$, $V_X$ and 
$\mathcal{C}$ of $(G',k')$ as obtained by Lemma~\ref{lem:modulator} in the case of 
\DKTED\ and that obtained after the last application of clique reduction rule in the 
case of \DED. 

\begin{lem}
  \label{lem:c1}
  $\sum_{C\in \mathcal{C}_1}|C| = O(k^2)$.
\end{lem}

\begin{proof}
  Let $\{v\}\in \mathcal{C}_1$. By Observation~\ref{obs:twox}, 
  $v$ must be adjacent to two vertices $x,y\in V_X$ such that $x$ and 
  $y$ are adjacent. Consider an edge $\{x,y\}\in X$. In the common neighborhood of $\{x,y\}$ there can be 
  at most $2k+1$ vertices $v$ with the property that $\{v\}\in \CC_1$, otherwise sunflower rule applies.
  Now, consider an edge $\{x,y\}\in E(G'[V_X]-X)$. In the common neighborhood of $\{x,y\}$
  there can be at most one vertex $v$  with the property that $\{v\}\in \CC_1$, otherwise there is
  an induced diamond edge-disjoint with $X$.
  Since there are $O(k)$ edges in $X$ and $O(k^2)$ edges in $E(G'[V_X]-X)$, we obtain the result.
\end{proof}

\begin{lem}
  \label{lem:bound-ordered}
  \begin{enumerate}[(i)]
    \item\label{item:lem-bound-ordered-aa} Consider any two vertices $x,y\in V_X$. 
      Let $\mathcal{C'}\subseteq \CC_{\geq 2}$
      such that for any $C\in \mathcal{C'}$, $x,y\in A_C$. If $\{x,y\}\in X$ then $|\mathcal{C}'|\leq 2k+1$. If $\{x,y\}\notin X$, 
      then $|\mathcal{C}'|\leq 1$.
    \item\label{item:lem-bound-ordered-ad} Consider any ordered pair of vertices $(x,y)$ in $V_X$ such that $x$ and $y$ are
      adjacent in $G'$. Let $\mathcal{C'}\subseteq \CC_{\geq 2}$
      such that for any $C\in \mathcal{C'}$, $x\in A_C$ and $y\in D_C$. 
      If $\{x,y\}\in X$ then $|\mathcal{C}'|\leq 2k+1$. If $\{x,y\}\notin X$, then $|\mathcal{C}'|= 0$.
  \end{enumerate}
\end{lem}

\begin{proof}
  (\ref{item:lem-bound-ordered-aa}).
  Let $C_a, C_b\in \mathcal{C'}$. By Observation~\ref{obs:intersection}(\ref{label:obs-intersection-1}), 
  $|C_a\cap C_b|\leq 1$.   If $v\in C_a\cap C_b$, then by Observation~\ref{obs:intersection}(\ref{label:obs-intersection-noedge}),
  there is no edge between $C_a\setminus\{v\}$ and $C_b\setminus \{v\}$. Hence, $\{x,v,a,b\}$ induces a diamond 
  where $a\in C_a\setminus\{v\}$ and $b\in C_b\setminus \{v\}$, which is edge disjoint with $X$, a contradiction.
  Hence $C_a\cap C_b=\emptyset$. Now, consider any two vertices $a\in C_a$ and $b\in C_b$. 
  Clearly, $\{x,y,a,b\}$ induces a diamond. Hence, $\{x,y\}$ must be an edge in $X$, otherwise the diamond
  is edge disjoint with $X$, a contradiction. Therefore, if $\{x,y\}\notin X$, $|\mathcal{C'}|\leq 1$.
  Now we consider the case in which $\{x,y\}\in X$. If $|\mathcal{C}'|\geq 2k+2$, we get at least $k+1$ diamonds where every two
  diamonds have the only edge intersection $\{x,y\}$. Then sunflower rule applies, which is a contradiction.

  (\ref{item:lem-bound-ordered-ad}).
  Let $\mathcal{C'}$ be the set of all $C\in \CC_{\geq 2}$ such that $x\in A_C$ and $y\in D_C$.
  Consider any two of them - $C_a$ and $C_b$. By Observation~\ref{obs:intersection}(\ref{label:obs-intersection-1}), 
  $|C_a\cap C_b|\leq 1$.   If $v\in C_a\cap C_b$, then by Observation~\ref{obs:intersection}(\ref{label:obs-intersection-noedge}),
  there is no edge between $C_a\setminus\{v\}$ and $C_b\setminus \{v\}$. Let $a\in C_a\setminus \{v\}$ and $b\in C_b\setminus \{v\}$.
  Then $\{x,v,a,b\}$ induces a diamond which is edge disjoint with $X$, a contradiction. Hence $C_a\cap C_b=\emptyset$.
  Let $a,a'\in C_a$ such that $a$ is adjacent to $y$. Then, if $\{x,y\}\notin X$, $\{x,a,a',y\}$ induces a diamond, which is 
  edge disjoint with $X$. Therefore, if $\{x,y\}\notin X$, then $|\mathcal{C'}|=0$. 
  Now we consider the case in which $\{x,y\}\in X$. If $|C'|\geq 2k+2$, we get at least $k+1$ diamonds 
  where   every two
  diamonds have the only edge intersection $\{x,y\}$. Then sunflower rule applies, which is a contradiction.
\end{proof}

\begin{lem}
  \label{lem:c2}
  \begin{enumerate}[(i)]
  \item\label{item:c2:ded}  For \DED,
    $\sum_{C\in \CC_{\geq 2}}|C|=O(k^3)$.
  \item\label{item:c2:dkted} For \DKTED,   
    $\sum_{C\in \CC_{\geq 2}}|C|=O(k^2)$.
  \end{enumerate}
\end{lem}

\begin{proof}
  (\ref{item:c2:ded}).
  Consider any two adjacent vertices $x,y\in V_X$. 
  Let $\mathcal{C}_{xy}'\subseteq \CC_{\geq 2}$ be such that 
  $x,y\in A_C$. Then by Lemma~\ref{lem:bound-ordered}(\ref{item:lem-bound-ordered-aa}), 
  if $\{x,y\}\in X$,
  then $|C'_{xy}|\leq 2k+1$ and if $\{x,y\}\notin X$, then $|C'_{xy}|\leq 1$. 
  Since there are at most $5k$ edges in $X$ and
  $O(k^2)$ edges in $G[V_X]\setminus X$, $\bigcup_{\{x,y\}\in E(G[V_X])} C'_{xy}$ 
  has at most $O(k)\cdot (2k+1)+O(k^2)=O(k^2)$ maximal
  cliques. Since every maximal clique has at most $4k$ 
  vertices (by Observation~\ref{obs:clique-bound}), the total
  number of vertices in those cliques is $O(k^3)$.

  Now, let $\mathcal{C'}_{xy}\subseteq \CC_{\geq 2}$ be such that 
  $x\in A_C$ and $y\in D_C$. Then by 
  Lemma~\ref{lem:bound-ordered}(\ref{item:lem-bound-ordered-ad}), if $\{x,y\}\in X$,
  then $|C'_{xy}|\leq 2k+1$ and if $\{x,y\}\notin X$, then $|C'_{xy}|= 0$. 
  Since there are at most $2\cdot 5k=10k$ ordered adjacent
  pairs of vertices in $X$,
  $\bigcup_{\{x,y\}\in E(G[V_X])} C'_{xy}$ has at most $O(k)\cdot (2k+1)$ maximal
  cliques. Since every maximal clique has at most $4k$ 
  vertices (by Observation~\ref{obs:clique-bound}), the total
  number of vertices in those cliques is $O(k^3)$.

  Since, by Observation~\ref{obs:twox}, 
  for every $C\in \mathcal{C}$, there exist two vertices $x\in A_C$ and $y\in A_C\cup D_C$,
  we have counted every $C\in \mathcal{C}_2\cup \mathcal{C}_{\geq 3}$. Hence 
  $\sum_{C\in \mathcal{C}_2\cup \mathcal{C}_{\geq 3}}|C|=O(k^3)$.

  (\ref{item:c2:dkted}).
  Since $G-V_X$ is $K_t$-free, every maximal clique in $\CC$ has at most $t-1$ vertices.
  Hence from the above arguments, we obtain the result.
\end{proof}  

\begin{thm}
  \label{thm:final-bound}
  Given an instance $(G,k)$ of \DED, the kernelization gives an instance $(G',k')$ such that $|V(G')| = O(k^3)$ 
  and $k'\leq k$ or declares that the instance is a no-instance. Similarly, given an instance 
  $(G,k)$ of \DKTED, the kernelization gives an instance $(G',k')$ such that $|V(G')|=O(k^2)$
  and $k'\leq k$ or declare that the instance is a no-instance.
\end{thm}

\begin{proof}
  None of the rules increases the parameter $k$. 
  Then, the theorem follows from Lemma~\ref{lem:c1} and Lemma~\ref{lem:c2} and the
  fact that $|V_X| = O(k)$.
\end{proof}

\section{Hardness Results}

%An $s$-diamond is defined as the graph $K_2\times (s+1)K_1$~\cite{AravindSS16parameterized}.
%We note that $1$-diamond is a diamond graph. An example is shown in Figure~\ref{fig:5d}.
In this section we prove that, for any fixed $s\geq 1$, 
\SDED\ is \NPC\ even for $K_4$-free graphs.
As a corollary, we obtain that, for any fixed $s\geq 1$ and $t\geq 4$, 
\SDKTED\ is \NPC. We also obtain that
these \NPC\ problems can be solved neither in subexponential time nor in 
parameterized subexponential time,
unless ETH fails. Further, we obtain similar results for \SDEE\ and \SDKTEE.

It is known that \VC\ is \NPC\ on sub-cubic graphs \cite{mohar2001face}.
It is also known that \VC\ is \NPC\ on triangle-free graphs \cite{poljak1974note, Yannakakis81edge} by the 
simple observation that a graph $G$ has a vertex cover of size 
at most $k$ if and only if the graph obtained from $G$ by 
sub-dividing every edge twice has a vertex cover of size at most $k+|E(G)|$.
Combining these two reductions implies that \VC\ is \NPC\ on triangle-free sub-cubic graphs.
Recently, Komusiewicz generalized this technique to obtain a general result, where
the result for triangle-free sub-cubic graphs comes as a special case \cite{komusiewicz2015tight}. 
It also gives that
\VC\ on triangle-free sub-cubic graphs cannot be solved in subexponential time, unless ETH fails.

\begin{pro}\cite{komusiewicz2015tight}
  \label{pro:vctf1}
  \VC\ on triangle-free sub-cubic graphs is \NPC. Further, the problem cannot be solved
  in time \SUBEX, unless ETH fails.
\end{pro}

Since, the reduction in \cite{komusiewicz2015tight} 
is not a linear parameterized reduction, we need to compose it
with the reduction from \TSAT\ to \VC\ to obtain a linear parameterized reduction.

\begin{lem}
  \label{lem:vctf2}
  \VC\ on triangle-free sub-cubic graphs cannot be solved in time \PSUBEX, unless ETH fails.
\end{lem}
\begin{proof}
  The reduction from \TSAT\ to \VC\ on sub-cubic graphs (see \cite{mohar2001face}) 
  gives an instance $(G,k)$ of \VC\ where $G$ is sub-cubic and 
  has $9m$ vertices and $12m$ edges and $k=5m$, where  $m$ is the number of 
  clauses of the input \TSAT\ instance. Now, replace every edge of $G$
  by a path of three edges (i.e., subdivide every edge twice) to obtain $G'$. 
  Now, $G'$ has $33m$ vertices and $36m$ edges.
  It is straight forward to verify that the input \TSAT\ instance is satisfiable if and 
  only if $G'$ has a vertex cover of size at most $17m$. Now, the statement follows from the 
  observation that the reduction is linear parameterized.
\end{proof}

Komusiewicz has also proved that \cite{komusiewicz2015tight} for any non-trivial hereditary property $\Pi$, 
\PIVD\ cannot be solved in time \SUBEX, unless ETH fails. The reduction is 
from a variant of \VC. We require only the case when $\Pi$ is `$K_{1,s}$-free'.
For the sake of completeness, we reiterate his proof tailored for our special case.

\begin{pro} \cite{komusiewicz2015tight}
  \label{pro:cvd}
  For any fixed $s\geq 2$, 
  \KOSVD\ on triangle-free graphs with maximum degree at most $s+2$
  is \NPC\ and can be solved neither in time \SUBEX\ nor in time
  \PSUBEX.
\end{pro}
\begin{proof}
  Let $(G,k)$ be an instance of \VC\ on triangle-free sub-cubic graphs. For every vertex $v_i$ in $G$,
  introduce an independent set $I_i$ of $s-1$ vertices and make all of them adjacent to $v_i$. 
  Let the resultant graph be $G'$. 
  Clearly, $G'$ is triangle-free and has degree at most $s+2$.
  
  Let $S$ be a vertex cover of size at most $k$ of $G$. Since $G-S$ is an independent set,
  $G'-S$ is a graph with degree at most $s-1$ and hence $K_{1,s}$-free. 
  Conversely, let $S'\subseteq V(G')$ be such that $|S'|\leq k$ and $G'-S'$ is $K_{1,s}$-free.
  For every vertex $u_i\in I_i$ such that $u_i\in S'$, replace $u_i$ by $v_i$ to obtain $S''$. 
  Clearly, $G'-S''$ is $K_{1,s}$-free and $|S''|\leq k$. It is straight forward to
  verify that $S''$ is an independent set of $G$, otherwise an edge $\{v_i,v_j\}$ in $G-S''$
  will cause a $K_{1,s}$ induced by $I_i\cup \{v_i,v_j\}$ in $G'-S''$, which is a contradiction. Now, the 
  statements follow from Proposition~\ref{pro:vctf1}, Lemma~\ref{lem:vctf2} and the fact that
  the reduction is both linear and linear parameterized.
\end{proof}

Now, we give a reduction from \KOSOVD\ on triangle-free graphs to 
\SDED\ (\SDEE) on $K_4$-free graphs.

\paragraph*{Reduction:} Let $s\geq 1$ be any fixed integer.
Let $(G,k)$ be an instance of \KOSOVD\ such that $G$
is triangle-free. Introduce a new vertex $w$ and make it adjacent
to all the vertices in $G$. Let the resultant graph be $G'$ and let
the reduced instance of \DED\ (\DEE) be $(G',k)$.

\begin{observation}
  \label{obs:red}
  For any fixed integer $s\geq 1$, 
  let $(G,k)$ be an instance of \KOSOVD\ on triangle-free graphs 
  and let $G'$ be obtained by the reduction described above. 
  \begin{enumerate}[(i)]
  \item\label{obs:red:0} $G'$ is $K_4$-free.
  \item\label{obs:red:1} 
    Let $D\subseteq V(G')$. Then, $D$ induces an $s$-diamond in $G'$
    if and only if $w\in D$ and $D\setminus \{w\}$ induces a $K_{1,s+1}$
    in $G$.
  \item\label{obs:red:2} 
    Let $S\subseteq V(G)$ and let $F$ be the set of all
    edges $\{w,v_i\}$, where $v_i\in S$.
    Then, $G-S$ is $K_{1,s+1}$-free if and
    only if $G' - F$ is $s$-diamond-free.
  \end{enumerate}
\end{observation}
\begin{proof}
  \begin{enumerate}[(i)]
    \item Implied by the fact that $G$ is triangle-free.
    \item
      Assume that $D$ induces an $s$-diamond in $G'$. Since $G$ is triangle-free,
      $w\in D$ and $w$ must be a vertex with degree $s+2$
      in the $s$-diamond $G'[D]$. Hence $D\setminus \{w\}$ induces a $K_{1,s+1}$ in $G$.
      The other direction is straight-forward. 
    \item 
      Let $G-S$ be $K_{1,s+1}$-free. By (\ref{obs:red:1}), every induced $s$-diamond in $G'$
      is formed by an induced $K_{1,s+1}$ in $G$ and $w$. Hence every $s$-diamond in $G'$
      is hit by $F$. Since $G'$ is $K_4$-free, no new $s$-diamond is created in $G'$ by
      deleting edges. Hence $G'-F$ is $s$-diamond-free.
      Conversely, let $G'-F$ be $s$-diamond-free. By (\ref{obs:red:1}), every
      induced $K_{1,s+1}$ in $G$ with $w$ forms an $s$-diamond in $G'$. Hence $S$ hits
      all induced $K_{1,s+1}$s in $G$. 
  \end{enumerate}
\end{proof}

\begin{thm}
  \label{thm:diamond-hard}
  For any fixed $s\geq 1$,
  \SDED\ and \SDEE\ are \NPC\ even on $K_4$-free graphs. Further, these
  \NPC\ problems can be solved neither in time \SUBEX\ nor in time \PSUBEX.
\end{thm}
\begin{proof}
  We reduce from \KOSOVD\ on triangle-free graphs. Let $(G,k)$ be an
  instance of \KOSOVD\ such that $G$ is triangle-free. 
  We apply the reduction described above to obtain
  an instance $(G',k)$ of \SDED~(\SDEE). 
  By Observation~\ref{obs:red}(\ref{obs:red:0}), $G'$ is $K_4$-free.
  We claim that $(G,k)$ is a yes-instance of \KOSOVD\ if and only if
  $(G',k)$ is a yes-instance of \SDED\ (\SDEE). 

  Let $S\subseteq V(G)$ be such that $|S|\leq k$ 
  and $G-S$ is $K_{1,s+1}$-free. Let $F\subseteq E(G')$ be defined 
  as the set of all edges $\{w,v_i\}$, where $v_i\in S$.
  Now, by Observation~\ref{obs:red}(\ref{obs:red:2}), $G'-F$ is $s$-diamond-free.

  Conversely, let $(G',k)$ be a yes-instance of \SDED~(\SDEE). 
  Let $T\subseteq E(G')$~($T\subseteq [V(G')]^2$) be such that $|T|\leq k$
  and $G'\triangle T$ is $s$-diamond-free. For every (potential) edge $\{v_i,v_j\}$
  of $G$ in $T$, replace $\{v_i,v_j\}$ with either $\{w,v_i\}$ or $\{w,v_j\}$ in $T$.
  Let the resultant set of edges be $T'$. Clearly, $|T'|\leq k$. 
  By \ref{obs:red}(\ref{obs:red:1}), every $s$-diamond in $G'$ contains $w$. Therefore,
  since $T$ hits every $s$-diamond in $G'$, $T'$ hits every $s$-diamond in $G'$. 
  Since $G'$ is $K_4$-free, no new $s$-diamond is created by deleting edges from $G'$.
  Hence, $G'-T'$ is $s$-diamond-free. Let $S$ be the set of all vertices
  $v_i$ such that $\{w,v_i\}$ is in $T'$. Then, by \ref{obs:red}(\ref{obs:red:1}),
  $G-S$ is $K_{1,s+1}$-free. Now, the statements follows from Proposition~\ref{pro:cvd} and 
  the observation that the reduction we give is both linear and linear parameterized.
\end{proof}

We observe that, in the proof of Theorem~\ref{thm:diamond-hard}, 
even for the editing problem, for every solution (the set of edges to be deleted and 
the set of edges to be added) of $(G',k)$,  has 
a corresponding solution which contains only the set of edges to be deleted. Hence we get the 
following corollary.

\begin{cor}
  For any fixed integers $s\geq 1$ and $t\geq 4$, \SDKTED\ and \SDKTEE\ are \NPC
  even on $K_4$-free graphs. Further, these problems 
  can be solved neither in time \SUBEX\ nor in time \PSUBEX.
\end{cor}

\section{Concluding Remarks}

Consider the graph $G$ in Figure~\ref{fig:difficult}.
A big circle denotes a clique of $k$ vertices. There are
$k$ of them. 
Outside the large 
cliques there are only four vertices which induces a 
diamond. One of those vertices is adjacent to all the 
vertices in the large cliques (thick edge denotes this) and another vertex in 
the diamond is adjacent to exactly one vertex from each large clique.
We observe that $G$ has $k^2+4$ vertices. 
None of our rules reduces the size of this graph and a $O(k)$ of 
such structures in a graph causes $O(k^3)$ vertices. We believe that
rules to tackle this structure is the key to obtain a smaller kernel for
\DED.

\begin{figure}[h]
  \centering
  \includegraphics[width=2.0in]{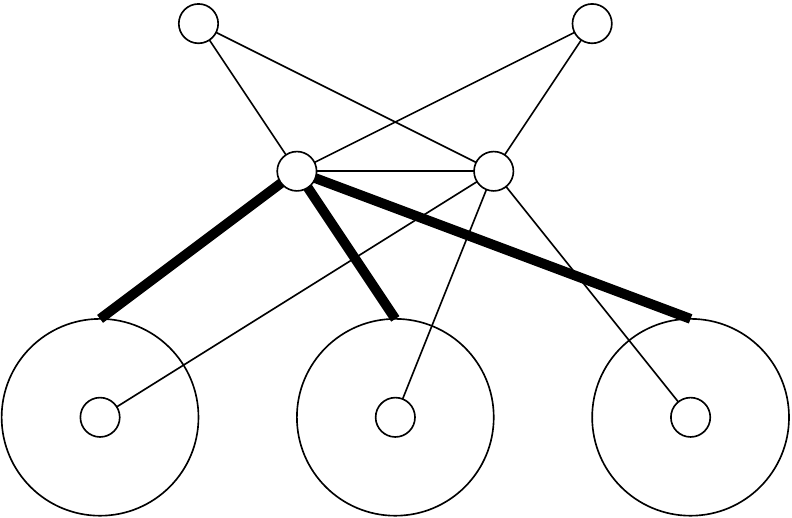}
  \caption{A structure to dismantle for a smaller kernel for \DED}
  \label{fig:difficult}
\end{figure}

\begin{op}
  Does \DED\ admit a kernel of $O(k^2)$ vertices?
\end{op}

We have proved the hardness and lower bounds for \SDED, for any fixed $s\geq 1$.
The vertex-split rule is safe to apply for these problems. Studying the 
structural properties of $s$-diamond-free graphs may help us to obtain a 
polynomial kernels for these problems when $s\geq 2$.

\begin{op}
  Does \SDED\ admit a polynomial kernel when $s\geq 2$?
\end{op}

Polynomial kernelization of \pname{Claw-free Edge Deletion} is considered as a difficult 
problem~\cite{CaiC15incompressibility, cygan2015polynomial}
in this area. One of the difficulties with this problem is that the characterization
of claw-free graphs is quite complicated. A paw graph is a graph obtained by adding 
an edge between two non-adjacent vertices in a claw. It is known that every component in a paw-free
graph is either triangle-free or complete multipartite \cite{Olariu88paw}. 
Can we use this to obtain a polynomial kernel for \pname{Paw-free Edge Deletion}?

\begin{op}
  Does \pname{Paw-free Edge Deletion} admit a polynomial kernel?
\end{op}

%%
%% Bibliography
%%

%% Either use bibtex (recommended), but commented out in this sample

\bibliography{diamond}

%% .. or use bibitems explicitely
%\newpage
%\appendix
%\input{appendix}

\end{document}